\documentclass[a4paper,12pt,table]{article}
\usepackage{amssymb,amsmath,latexsym,color,xcolor,amsthm,authblk,mathpazo,palatino,bbding,setspace,datetime,breakcites,hyphenat,microtype,diagbox,tikz,subcaption,graphicx,titlesec}
\usepackage[round,authoryear]{natbib}
\usepackage{hyperref,theoremref}
\usepackage{cleveref}
\usepackage{xspace}
\usepackage{comment}
\usepackage{xifthen}
\usepackage[titletoc,title]{appendix}
\usepackage[shortlabels]{enumitem}
\usepackage{graphicx}
\usepackage{float}
\usepackage{amssymb}
\restylefloat{table}
\everymath{\displaystyle}
\usepackage{pstricks}
\usepackage{blkarray}

\theoremstyle{plain}

\definecolor{armygreen}{rgb}{0.29, 0.8, 0.13}
\definecolor{auburn}{rgb}{0.43, 0.21, 0.1}
\definecolor{burgundy}{rgb}{0.5, 0.0, 0.13}
\definecolor{medium red}{rgb}{.490,.298,.337}
\definecolor{dark red}{rgb}{.235,.141,.161}
   
\hypersetup{
	colorlinks = true,
	linkcolor = {burgundy},
	citecolor = {burgundy}, 		
	linkbordercolor = {white},
}

\let\OLDthebibliography\thebibliography
\renewcommand\thebibliography[1]{
	\OLDthebibliography{#1}
	\setlength{\parskip}{0pt}
	\setlength{\itemsep}{0pt plus 0.1ex}
}

\DeclareFontFamily{U}{mathx}{\hyphenchar\font45}
\DeclareFontShape{U}{mathx}{m}{n}{<-> mathx10}{}
\DeclareSymbolFont{mathx}{U}{mathx}{m}{n}
\DeclareMathAccent{\widebar}{0}{mathx}{"73}

\titleformat{\section}[block]{\normalfont\scshape\large\filcenter}{\thesection}{1em}{}
\titleformat{\subsection}{\normalfont\scshape\large}{\thesubsection}{1em}{}
\titleformat{\subsubsection}{\normalfont\scshape\large}{\thesubsubsection}{1em}{}

\newtheorem{theorem}{Theorem}
\newtheorem*{theorem*}{Theorem}

\newtheorem{claim}{Claim}
\newtheorem{lemma}{Lemma}

\newtheorem*{claim*}{Claim}

\theoremstyle{definition}
\newtheorem{definition}{Definition}
\newtheorem{example}{Example}
\theoremstyle{remark}
\newtheorem{remark}{\textsc{Remark}}

\makeatletter
\newcommand*\bigcdot{\mathpalette\bigcdot@{.5}}
\newcommand*\bigcdot@[2]{\mathbin{\vcenter{\hbox{\scalebox{#2}{$\m@th#1\bullet$}}}}}
\makeatother
\newcommand{\objects}{\ensuremath{X}\xspace}
\newcommand{\randoma}{\ensuremath{A}\xspace}
\newcommand{\randomaof}[2]{
    \ifthenelse{\isempty{#1}}
    {
        \ifthenelse{\isempty{#2}}
        {\ensuremath{\randoma_{i\boldsymbol{\bigcdot}}}\xspace}
        {\ensuremath{\randoma_{i#2}}\xspace}
    }
    {
        \ifthenelse{\isempty{#2}}
        {\ensuremath{\randoma_{#1\boldsymbol{\bigcdot}}}\xspace}
        {\ensuremath{\randoma_{#1#2}}\xspace}
    }
}

\newcommand{\rendow}[2]{
    \ifthenelse{\isempty{#1}}
        {
            \ifthenelse{\isempty{#2}}
            {\ensuremath{E_{i\boldsymbol{\bigcdot}}}\xspace}
            {\ensuremath{E_{i#2}}\xspace}
        }
        {
            \ifthenelse{\isempty{#2}}
            {\ensuremath{E_{#1\boldsymbol{\bigcdot}}}\xspace}
            {\ensuremath{E_{#1#2}}\xspace}
        }
}
\newcommand{\dendow}[1]{
    \ifthenelse{\isempty{#1}}
    {\ensuremath{x_{i}}\xspace}
    {\ensuremath{x_{#1}}\xspace}
}
\newcommand{\rrule}{\ensuremath{\varphi}\xspace}
\newcommand{\rruleof}[3]{
    \ifthenelse{\isempty{#1}}
    {
        \ifthenelse{\isempty{#2}}
        {\ensuremath{\rrule_{i\boldsymbol{\bigcdot}}(#3)}\xspace}
        {\ensuremath{\rrule_{i#2}(#3)}\xspace}
    }
    {
        \ifthenelse{\isempty{#2}}
        {\ensuremath{\rrule_{#1\boldsymbol{\bigcdot}}(#3)}\xspace}
        {\ensuremath{\rrule_{#1#2}(#3)}\xspace}
    }
}
\newcommand{\without}[1]{
    \ifthenelse{\isempty{#1}}
    {\ensuremath{P_{-i}}}
    {\ensuremath{P_{-#1}}}
}
\newcommand{\prefof}[3]{
    \ifthenelse{\isempty{#3}}
    {
        \ifthenelse{\isempty{#1}}
            {
                \ifthenelse{\isempty{#2}}
                {\ensuremath{P_i(1)}\xspace}
                {\ensuremath{P_i(#2)}\xspace}
            }
            {
                \ifthenelse{\isempty{#2}}
                {\ensuremath{P_{#1}(1)}\xspace}
                {\ensuremath{P_{#1}(#2)}\xspace}
            }
    }
    {
        \ifthenelse{\isempty{#1}}
            {
                \ifthenelse{\isempty{#2}}
                {\ensuremath{P_i^{#3}(1)}\xspace}
                {\ensuremath{P_i^{#3}(#2)}\xspace}
            }
            {
                \ifthenelse{\isempty{#2}}
                {\ensuremath{P_{#1}^{#3}(1)}\xspace}
                {\ensuremath{P_{#1}^{#3}(#2)}\xspace}
            }
    }
}

\newcommand{\suml}{\sum\limits_}

\newdateformat{monthyeardate}{\monthname[\THEMONTH], \THEYEAR}

\title{\textsc{{A characterization of strategy-proof probabilistic assignment rules}}\thanks{{\textit{The authors are listed in alphabetical order of their last names.}
An earlier version of this paper, circulated as ``On Probabilistic Assignment Rules'' (\cite{gogulapati2025probabilisticassignmentrules}), contained preliminary results on this topic.} }}
\author[]{Sai Praneeth Donthu\thanks{Indian Institute of Technology Kanpur; Email: praneeth.donthu@gmail.com} }
\author[]{Souvik Roy\thanks{Indian Statistical Institute, Kolkata; Email: gametheory.souvik@gmail.com} } 
\author[]{Soumyarup Sadhukhan\thanks{Indian Institute of Technology Kanpur; Email: soumyarups@iitk.ac.in} }
\author[]{Gogulapati Sreedurga\thanks{Indian Institute of Technology Hyderabad; Email: gogulapatis@alum.iisc.ac.in} }
\affil[]{}
\date{\monthyeardate\today}

\begin{document}
	
	\maketitle
	
	\begin{abstract}\singlespacing
We study the classical \emph{probabilistic assignment problem}, where finitely many indivisible objects are to be probabilistically or proportionally assigned among an equal number of agents. Each agent has an initial deterministic endowment and a strict preference over the objects. While the deterministic version of this problem is well understood, most notably through the characterization of the Top Trading Cycles (TTC) rule by \citet{ma1994strategy}, much less is known in the probabilistic setting. Motivated by practical considerations, we introduce a weakened incentive requirement, namely \emph{SD-top-strategy-proofness}, which precludes only those manipulations that increase the probability of an agent’s top-ranked object.

Our first main result shows that, on any \emph{free pair at the top (FPT)} domain \citep{sen2011gibbard}, the TTC rule is the unique probabilistic assignment rule satisfying \emph{SD–Pareto efficiency}, \emph{SD–individual rationality}, and \emph{SD–top-strategy-proofness}. We further show that this characterization remains valid when Pareto efficiency is replaced by the weaker notion of \emph{SD–pair efficiency}, provided the domain satisfies the slightly stronger \emph{free triple at the top (FTT)} condition \citep{sen2011gibbard}. Finally, we extend these results to the ex post notions of efficiency and individual rationality. 

Together, our findings generalize the classical deterministic results of \citet{ma1994strategy} and \citet{ekici2024pair} along three dimensions: extending them from deterministic to probabilistic settings, from full strategy-proofness to top-strategy-proofness, and from the unrestricted domain to the more general FPT and FTT domains.

		
		\noindent 
				
		\vspace{4mm}
		
        
\noindent MSC2010 Subject Classification: 91B14, 91B03, 91B68
		
		\vspace{4mm}
		
	\end{abstract}
	

	\maketitle
	
{
	\def\OldComma{,}
	\catcode`\,=13
	\def,{%
		\ifmmode%
		\OldComma\discretionary{}{}{}%
		\else%
		\OldComma%
		\fi%
	}%


\section{Introduction}\label{se_intro}
\subsection{Description and motivation of the problem}

We consider the classical assignment problem with initial endowments in the probabilistic setting. In an assignment problem, a finite set of objects has to be allocated among a finite set of agents (with the same cardinality)  based on their (strict) preferences over the objects. Moreover, the agents have endowments over the objects. Practical applications of this model can be thought of as house allocation among existing tenants, course allocation among the faculty members who have endowments for the courses they taught last time, office space assignment, parking space assignment, and many others. 

A deterministic assignment rule (henceforth, DAR) is a function from the set of reported preferences to the set of assignments, where an assignment is a 1-1 mapping between the set of agents and the set of objects. Some desirable properties of deterministic rules are strategy-proofness, Pareto-efficiency, individual rationality, pair-efficiency, etc.; a deterministic rule is strategy-proof if no agent has an incentive to misreport her preference to get a better outcome. 
Pareto efficiency requires that there exists no group of agents who can reshuffle their assignments, under the deterministic rule, in such a way that all agents in the group are weakly better off, with at least one agent strictly better off. A weaker notion is that of pair efficiency, which imposes the same requirement but restricts attention only to groups consisting of two agents.
Finally, a deterministic rule is individually rational if every agent is guaranteed to get a better outcome than her initial endowment according to her preference.

\cite{ma1994strategy} shows that on the unrestricted domain (when all possible orderings over the objects are feasible), a deterministic rule satisfies strategy-proofness, Pareto-efficiency, and individual rationality if and only if it is the TTC rule, a deterministic rule first introduced in \cite{shapley1974cores} in the context of the object reallocation problem. \cite{svensson1994queue}, \cite{anno2015short}, and \cite{sethuraman2016alternative} 
provided shorter proofs of this result. Recently, \cite{ekici2024pair} shows that the same result holds even if we replace Pareto-efficiency with a much weaker condition called pair-efficiency. 

In the probabilistic setup, instead of a deterministic rule, a probabilistic assignment rule (henceforth, PAR)  is considered, which assigns a bi-stochastic matrix at every instance of reported preferences of the agents. Here, the order of a bi-stochastic matrix is the common cardinality of the agent set and the object set, and each row of the matrix denotes the probabilities of assigning the objects to a particular agent, whereas the columns represent the probabilities of assigning a specific object to different agents. In the mechanism design literature, it is well-known that probabilistic rules are better in terms of fairness consideration when compared to their deterministic counterpart. For probabilistic rules, the corresponding desirable properties like strategy-proofness, Pareto-efficiency, individual rationality, and pair-efficiency can be defined using the first-order stochastic dominance.

In this paper, we investigate the structure of probabilistic rules satisfying the desired properties, namely, efficiency (both Pareto and pair), individual rationality, and strategy-proofness, under the assumption that the initial endowment is deterministic. The concept of pair efficiency was introduced only recently by \citet{ekici2024pair}. This notion of efficiency is considerably weaker than the classical Pareto efficiency in the deterministic setting (see Example~1 in \cite{ekici2024pair}), and it becomes even weaker in the probabilistic framework (see Example~1). We also consider a weaker version of the standard SD-strategy-proofness notion, which allows for better real-life applicability of our results. To the best of our knowledge, this is the first work that examines these notions in a probabilistic environment.
In what follows, we elaborate on the motivation for weakening the notion of strategy-proofness and discuss the relevance of assuming deterministic initial endowments in the probabilistic context.

While Pareto-efficiency, pair efficiency, and individual rationality have straightforward extensions in the probabilistic setup using first-order stochastic dominance, many extensions of the notion of strategy-proofness are possible (see \cite{sen2011gibbard, aziz2015generalizing, chun2020upper}). The most used one in the literature (\cite{gibbard1977manipulation}) assumes that an agent will manipulate if, by misreporting, she can be better off for at least one upper contour set.  No doubt, this is quite a strong condition and may not be appropriate in a practical scenario, as in reality, agents might not care for all upper contour sets. Keeping this point in mind, we consider a weakening of strategy-proofness that we call as SD-top-strategy-proofness.
 Top-strategy-proofness assumes the agents only care about manipulation if they can increase the probability of their top alternative. The deterministic version of top-strategy-proofness implies agents manipulate only if, by misreporting, they can ensure their top alternative as the outcome. Thus, top-strategy-proofness is the minimal form of strategy-proofness one may think of.

    Throughout this paper, we assume that the initial endowment is deterministic. This assumption is consistent with the observation that, even when the current endowment is obtained through a probabilistic rule in the past, the very purpose of such a rule is to generate a deterministic realization according to the prescribed probability distribution. Consequently, the outcome after implementation is always deterministic. It is important to emphasize, however, that this does not diminish the relevance of probabilistic rules in ensuring fairness.

\subsection{Contribution of the Paper}\label{se_1.2}

We have four main results in this paper. In Theorem \ref{thm_1}, we provide a characterization of all SD-Pareto efficient, SD-individually rational, and SD-top-strategy-proof probabilistic rules on any \textit{free pair at the top} (FPT) domain by showing that the TTC rule is the unique rule satisfying these properties. A domain is FPT (\cite{sen2011gibbard}) if for any two objects $a$ and $b$, there exists a preference that places $a$ at the best and $b$ at the second-best ranks. FPT domains form a broad class that includes the unrestricted domain and can be substantially smaller in size. While the unrestricted domain requires $n!$ preference orderings, an FPT domain may be constituted with as few as $n(n-1)$ preferences. The linked domain, introduced by \citet{aswal2003dictatorial}, serves as an important example of an FPT domain. Our result generalizes \cite{ma1994strategy}'s result in three directions: first, from the deterministic setup to the probabilistic setup; second, from strategy-proofness to top-strategy-proofness; and third, from the unrestricted domain to any FPT domain. 

Next, we move to SD-pair efficiency and show that the same result continues to hold if we weaken Pareto efficiency to pair efficiency under stochastic dominance, but strengthen the domain slightly by requiring a \textit{free triple at the top} (FTT) domain instead of the FPT domain. An FTT domain requires that for every triplet of objects $a, \;b$, and $c$, there exists a preference with $a$ as the best, $b$ as the second-best, and $c$ as the third-best object in the preference. Formally, Theorem \ref{thm_2} shows that on any FTT domain, a probabilistic rule is SD-pair efficient, SD-individually rational, and SD-top-strategy-proof if and only if it is the TTC rule. As in the previous paragraph, this result generalizes Theorem 1 of \cite{ekici2024pair} in three important ways: first,  it extends from a deterministic framework to a probabilistic one;  second, from SD-strategy-proofness to SD-top-strategy-proofness; and third, from the unrestricted domain to any FPT domain. Finally, Theorem \ref{thm_3} and Theorem \ref{thm_4} consider the ex post notions of efficiency and individual rationality (\cite{abdulkadirouglu1999house}) and show that under these ex post notions, the conclusions of Theorem \ref{thm_1} and Theorem \ref{thm_2} also hold.

\subsection{Related Literature}
There have been quite a few works in the probabilistic setup of the assignment problem with initial endowments. \cite{athanassoglou2011house} show that when initial endowments are probabilistic (also known as fractional endowments), SD-strategy-proofness, SD-Pareto efficiency, and SD-individual rationality are together incompatible when there are at least four agents and objects. In a later paper, \cite{aziz2015generalizing} strengthens the result by weakening SD-strategy-proofness with weak-SD-strategy-proofness. In both their proofs, they started from some specific fractional endowments and arrived at an impossibility.\footnote{See the proofs of Theorem 3 in \cite{athanassoglou2011house} and Theorem 8 in \cite{aziz2015generalizing} for details.} To the best of our knowledge, there is not much known in the probabilistic setup when the feasible set of initial endowments is restricted. 

In a slightly different approach to the probabilistic assignment problem, one of the first papers is \cite{bogomolnaia2001new}. They introduce a new probabilistic rule called Probabilistic Serial (PS rule) and show that it satisfies some nice properties over the Random Priority rule. They also show an important impossibility result that there is no probabilistic rule satisfying  SD-efficiency, equal treatment of equals, and  SD-strategy-proofness. Equal treatment of equals ensures that two agents with the same preference get the same outcome. Later on, \cite{bogomolnaia2012probabilistic} characterize the PS rule in terms of  SD-efficiency, SD-no-envy, and bounded invariance, where SD-envy-free requires that every agent prefers their share over others. Bounded invariance is a weaker notion of SD-strategy-proofness.
\cite{chun2020upper} strengthen the impossibility result of \cite{bogomolnaia2001new} by weakening SD-strategy-proofness to upper-contour strategy-proofness, which only requires that if the upper-contour sets of some objects are the same in two preference relations, then the sum of probabilities assigned to the objects in the two upper-contour sets should be the same. 

\cite{yilmaz2009random} considers probabilistic assignment problems under weak preferences. Their main contribution is a recursive solution for the weak preference domain that satisfies individual rationality, ordinal efficiency and no justified-envy. No justified-envy views an assignment as unfair if an agent does not prefer his consumption to another agent's consumption, and the assignment obtained by swapping their consumptions respects the individual rationality requirement of the latter agent. 

Another paper that considers probabilistic allocations with a deterministic initial allocation is \cite{abdulkadirouglu1999house}. In their setting, there are existing tenants as well as new applicants. They introduce a deterministic rule called top trading cycles with fixed priority, which boils down to the TTC rule if there are only existing tenants. They also consider probabilistic allocations, which are a convex combination of these deterministic rules and show that these convex combination rules are strategy-proof, efficient, and individually rational.

We organize the paper in the remaining sections as follows: Section \ref{sec: prelims} introduces the model and the required definitions. Section \ref{sec: ttc} describes the TTC rule in an algorithmic way. In Section \ref{sec: sd}, we present our two main results under stochastic dominance notions. Further, Section \ref{sec: exp} provides the results on ex post notions of efficiency and individual rationality. Finally, in Section \ref{sec: fd}, we discuss some future directions of this work. We put all the missing proofs in the Appendix.



\section{Preliminaries}\label{sec: prelims}
Let $N=\{1,\ldots,n\}$ be a finite set of agents. Except otherwise mentioned, $n \geq 2$. Let $\objects = \{x_1,\ldots,x_n\}$ be a finite set of objects. A  reflexive, anti-symmetric, transitive, and complete binary relation (also called a linear order) on the set \objects is called a preference on \objects. We denote by $\mathcal{P}$ the set of all preferences on \objects, also called the unrestricted domain. For $P \in \mathcal{P}$ and  $x,y \in \objects$, $xPy$ is interpreted as ``$x$ is as good as (that is, weakly preferred to) $y$ according to $P$". Since $P$ is complete and antisymmetric, for distinct $x$ and $y$, we have either $xPy$ or $yPx$, and in such cases, $xPy$ implies $x$ is strictly preferred to $y$.  For ease of writing, we sometimes write a preference $P$ as $P\equiv xyz\cdots$, implying $x$ is the top-ranked object in $P$, and $y,z$ are the second and third-ranked objects, respectively, in $P$. For $P \in \mathcal{P}$ and $x \in \objects$, the \textit{upper contour set} of $x$ at $P$, denoted by $U(x,P)$, is defined as the set of objects that are as good as $x$ in $P$, i.e., $U(x,P)=\{y \in \objects \mid yPx\}$.\footnote{Observe that $x\in U(x,P)$ by reflexivity.} We let $\mathcal{D} \subseteq \mathcal{P}$ denote a domain (of admissible preferences). An element $P_N = (P_1, \ldots, P_n) \in \mathcal{D}^n$ is referred to as an (admissible) preference profile in the domain $\mathcal{D}$.

\subsection{Probabilistic assignment rules and their properties}

A probabilistic assignment \randoma is a bi-stochastic matrix of order $n$, i.e., \randoma is an $n \times n$ matrix in which every entry is in between $0$ and $1$ and every row and column sums are  $1$ (i.e., $0 \leq \randomaof{i}{j} \leq 1$, $\suml{j \in [n]}{\randomaof{i}{j}}=1$, and $\suml{i \in [n]}{\randomaof{i}{j}}=1$).\footnote{For $p\in \mathbb{N}$, $[p]=\{1,\ldots,p\}$.} We denote by $\mathcal{\randoma}$ the set of all probabilistic assignments. When each entry of a probabilistic assignment \randoma is either zero or one (but not a proper fraction), then it is called a deterministic assignment. Note that a deterministic assignment is a permutation matrix.

The rows of a probabilistic assignment correspond to the agents, and the columns correspond to the objects. For an agent $i$ and object $x$, the value $\randomaof{i}{j}$ denotes the probability with which agent $i$ receives object $x_j$. Using standard matrix notations, we write $\randomaof{i}{\boldsymbol{\bigcdot}} \in  \Delta(\objects)$ to denote the probability distribution in the $i^{\text{th}}$ row of \randoma.\footnote{We use the notation  $\Delta (S)$ to denote the set of probability distributions on a finite set $S$.} Likewise, $\randomaof{\boldsymbol{\bigcdot}}{j} \in \Delta(N)$ denotes the probability distribution in the column corresponding to the object $x_j$. Further, with a slight abuse of the above notations, for any set $S \subseteq \objects$, let $\randomaof{i}{S}$ denote the total probability of the elements in $S$ in the $i^{\text{th}}$ row of \randoma. That is, $\randomaof{i}{S}  = \suml{j \in S}{\randomaof{i}{j}}$. Likewise, for any $S \subseteq N$, we define $\randomaof{S}{j}$ to be the total probability of the elements in $S$ in the $j^{\text{th}}$ column of \randoma.

In this paper, we assume that there is an initial deterministic assignment of the objects to the individuals. Let $E$ (a permutation matrix) denote that initial assignment, also referred to as \textbf{initial endowment}. For notational simplicity, we assume agent $i$ has the object $x_i$ as her initial assignment. This means essentially $E$ is the identity matrix of order $n$. Further, with a slight abuse of notation, sometimes, for convenience, we write the initial endowment $E$ as a vector, i.e., $E=(x_1,\ldots,x_n)$.

A \textbf{probabilistic assignment rule} (henceforth, a PAR) is a function $\rrule: \mathcal{D}^n \to \mathcal{\randoma}$. The entry $\rruleof{i}{x_j}{P_N}$ represents the probability with which agent $i$ receives object $x_j$ at the profile $P_N$. A PAR $\rrule$ is said to be \textbf{deterministic assignment rule} (henceforth, a DAR and typically denoted by $f$) if for every $P_N \in \mathcal{D}^n$, $\rrule(P_N)$ is a deterministic assignment.

Given a preference $P_i \in \mathcal{P}$ of the agent $i$, we say that the agent $i$ \textbf{weakly prefers} a probability distribution $\lambda \in \Delta (\objects)$ over a distribution $\lambda' \in \Delta(\objects)$, if for every $x_j \in \objects$, we have $\lambda(U(x,P_i)) \geq \lambda'(U(x,P_i))$. Likewise, agent $i$ \textbf{strictly prefers} $\lambda$ over $\lambda'$ if it weakly prefers $\lambda$ over $\lambda'$ and there exists some $x \in \objects$ such that $\lambda(U(x,P_i)) > \lambda'(U(x,P_i))$. We write $\lambda \succeq_{P_i} \lambda'$ to indicate that $i$ weakly prefers $\lambda$ over $\lambda'$, and $\lambda \succ_{P_i} \lambda'$ to indicate that $i$ strongly prefers $\lambda$ over $\lambda'$.

We now introduce the standard notion of strategy-proofness in the probabilistic setting, called SD-strategy-proofness, which was first introduced in \cite{gibbard1977manipulation}.

\begin{definition}(SD-Strategy-proofness)\label{def: sp}
A PAR $\rrule$ is \textbf{SD-strategy-proof} if for every $P_N \in \mathcal{D}^n$, every agent $i \in N$, and every $P_i' \in \mathcal{D}$, $$\rruleof{}{}{P_N} \succeq_{P_i} \rruleof{}{}{P_i',\without{}}.$$
\end{definition}

As mentioned in Section \ref{se_intro}, below we formally define a weaker notion of SD-strategy-proofness that we work with in the results. 

\begin{definition}[SD-top-strategy-proofness]\label{def: topsp}
A PAR $\rrule$ is \textbf{SD-top-strategy-proof} if for every $P_N \in \mathcal{D}^n$, every agent $i \in N$, and every $P_i' \in \mathcal{D}$, $$\rruleof{}{\prefof{}{}{}}{P_N} \geq \rruleof{}{\prefof{}{}{}}{P_i',\without{}}.$$
\end{definition}

Both these properties can be defined similarly for the deterministic assignment rules, and we refer to them as strategy-proofness and top-strategy-proofness throughout the rest of the paper. 


\section{Top Trading Cycles (TTC) Rule}\label{sec: ttc}

In this section, we formally define the well-known Top Trading Cycles (TTC) rule. The standard formulation of the TTC rule assumes a deterministic initial endowment of objects among agents. This setting is consistent with the framework adopted in our model, wherein initial endowments are also deterministic. 

Recall that each agent~$i$ is initially endowed with a distinct object~$x_i$. We use $\objects_S$ to denote the set of endowments of all the agents in $S \subseteq N$. We additionally need the following terminology. For a preference $P$ on \objects and a subset $Y$ of \objects, we write $P|_Y$ to refer to the restriction of $P$ to $Y$, that is, $P|_Y$ is the preference on $Y$  such that for all $x,y \in Y$, $xP|_Yy$ if and only if $xPy$. For a preference profile $P_N$, we write $P_N|_Y$ to refer to the profile in which all the preferences in $P_N$ are restricted to $Y$.


We now introduce a particular type of graph. Let $S \subseteq N$ and $P_i$ be a preference on $X_S$ for an agent $i \in S$. The graph $\mathcal{G}(P_S)$ is defined as follows: the set of nodes is $S$ and there is an edge $(i,j)$ if and only if $P_i=x_j\cdots$. As $N$ is a finite set, the graph $\mathcal{G}(P_N)$ will have at least one cycle. 

We are now ready to define the TTC rule. Consider a preference profile $P_N$. The TTC rule is a deterministic assignment rule whose step-by-step description at this profile is as follows:
\begin{enumerate}
    \item \label{nodes} Round 1: Let $\objects^1=\objects$,  $N^1=N$, and $P^1_{N^1}=P_N$. Consider the graph  $\mathcal{G}^1(P^1_{N^1})$. 
    Let $C=(i_1,\ldots, i_k,i_1)$ be the smallest cycle in the graph. Assign $\dendow{j+1}$ to agent $j$ for all $j \in \{1,\ldots,k\}$ (where $i_{k+1}=i_1$). 
    
    \item  \label{intialize}  Round 2: Let  ${N}^2$ be the set of remaining agents (that is, the agents who did not belong to any cycle in the first round). Let $\objects^2=X_{{N}^2}$ and  ${P}^2_{{N}^2}$ be the reduced preference profile of the agents in $N_2$ to the set of objects in $\objects^2$, that is, ${P}^2_i=P_i^1|_{\objects^2}$. Now, consider the graph $\mathcal{G}^2(P^2_{N^2})$ and repeat Round 1.

    
    This continues till all the objects are allocated to some agents. Note that this always results in a deterministic assignment.
    
\end{enumerate}

The following theorem is due to \cite{ma1994strategy}. Later on, shorter proofs are found by \cite{svensson1994queue}, \cite{anno2015short}, and \cite{sethuraman2016alternative}.
\begin{theorem*}
On the unrestricted domain, TTC is the unique deterministic assignment rule that satisfies Pareto-efficiency, individual rationality, and strategy-proofness. 
\end{theorem*}


\section{Results under Stochastic Dominance notions}\label{sec: sd}

As mentioned in Section~\ref{se_intro}, the primary objective of this paper is to investigate the structure of desirable probabilistic assignment rules (PARs). Broadly, there exist two principal approaches to defining properties for PARs, namely, the stochastic dominance (SD) approach and the ex-post approach. While the SD approach is standard in the literature on probabilistic mechanism design, the ex-post approach, employed in~\cite{abdulkadirouglu1999house}, is based on convex combinations of the corresponding deterministic properties. We adopt the SD framework in this section and turn to the ex-post framework in the subsequent section.

We consider two notions of SD-efficiency, both of which extend the classical definition of efficiency to probabilistic outcomes. The first one is \textit{SD-Pareto efficiency}, where a probabilistic outcome is considered more efficient if there exists a subset of agents who strictly prefer it, while all remaining agents are indifferent.\footnote{This notion is also known in the literature as ordinal efficiency.} 
The second notion, which we refer to as \emph{SD-pair efficiency}, is a weaker criterion. An outcome is SD-pair efficient over another if exactly two agents strictly prefer it over the other, and all other agents have the same assignment in both. Clearly, SD-pair efficiency is a special case of  SD-Pareto efficiency. Consequently, if an outcome is SD-Pareto efficient, it is also pairwise Pareto efficient. However, the converse does not necessarily hold. It is worth emphasizing that, while SD-Pareto efficiency is already well established in the literature, to the best of our knowledge, the notion of SD-pair efficiency is introduced for the first time in this paper.

\subsection{Results under SD-Pareto Efficiency} 
We first present a formal definition of SD-Pareto efficiency.

\begin{definition}(SD-Pareto efficiency)\label{def: pe} A probabilistic assignment \randoma is SD-Pareto dominated by another probabilistic assignment $\randoma'$ at a profile $P_N$ if $$\randoma'_{i\boldsymbol{\bigcdot}} \succeq_{P_i} \randoma_{i\boldsymbol{\bigcdot}} \text{ for all } i\in N \text{ and } \randoma'_{j\boldsymbol{\bigcdot}} \succ_{P_j} \randoma_{j\boldsymbol{\bigcdot}} \text{ for some } j\in N.$$  A probabilistic assignment is said to be SD-Pareto efficient if it is not SD-Pareto dominated by some other probabilistic assignment. 
A PAR $\rrule$ is \textbf{SD-Pareto-efficient} if for all $P_N\in \mathcal{D}^n$, $\rrule(P_N)$ is  SD-Pareto efficient at $P_N$. 	
\end{definition}

Next, we formally define the stochastic dominance version of individual rationality.

\begin{definition}(SD-individual Rationality)\label{def: ir}
 A PAR $\rrule$ is \textbf{SD-individually rational} if for every $P_N \in \mathcal{D}^n$ and every agent $i \in N$, we have $$\rruleof{}{}{P_N} \succeq_{P_i} \rendow{}{} \;\;.$$
\end{definition}

Similarly to before, when we talk about SD-Pareto efficiency and SD-individual rationality in the deterministic setting, we call them Pareto efficiency and individual rationality. To state our main result in this section, we define a particular type of domain that was introduced in \cite{sen2011gibbard}.

\begin{definition}
    A domain $\mathcal{D}$ is a Free Pair at the Top (FPT) domain if for all distinct $x,y\in X$ there exists $P\in \mathcal{D}$ such that $P\equiv xy\cdots$.
\end{definition}

Note that the unrestricted domain $\mathcal{P}$ is an FPT domain. From \cite{ma1994strategy}'s result, we know that, on the unrestricted domain, the TTC is the unique Pareto efficient, individually rational, and strategy-proof DAR when the endowments are deterministic. In this section, we extend this result to probabilistic assignments and prove that even in such a setting, on any FPT domain, TTC continues to be the only SD-Pareto-efficient, SD-individually rational, and SD-top-strategy-proof assignment rule starting from deterministic initial endowments.

\begin{theorem}\label{thm_1}
    Let $\mathcal{D}$ be an FPT domain. Then, a PAR on $\mathcal{D}$ is  SD-Pareto efficient, SD-individually rational, and SD-top-strategy-proof if and only if it is the TTC rule. 
\end{theorem}

As a corollary of Theorem \ref{thm_1}, we get \cite{ma1994strategy}'s result stated in the previous section.

\subsection{Results under SD-pair efficiency}
The notion of pair efficiency was introduced in \cite{ekici2024pair}. We extend this in the probabilistic setting using stochastic dominance. The notion of SD-pair-efficiency requires that no pair of agents can be made strictly better off with respect to stochastic dominance by reallocating their assigned outcomes among themselves, while the allocations of all other agents remain unchanged. Below we formally define it.

\begin{definition}(SD-pair-efficiency)\label{def: paire} 
A probabilistic assignment \randoma is SD-pair dominated at a profile $P_N$ by another probabilistic assignment \randoma'  if there exist $i,j\in N$ such that
\begin{enumerate}
\item [(i)] $\randoma'_{l\boldsymbol{\bigcdot}} \succeq_{P_l} \randoma_{l\boldsymbol{\bigcdot}}$  for all  $l\in \{i,j\}$ with  $\randoma'_{h\boldsymbol{\bigcdot}} \succ_{P_h} \randoma_{h\boldsymbol{\bigcdot}}$  for some  $h\in \{i,j\}$, and
    \item [(ii)] $\randoma_{k\boldsymbol{\bigcdot}}=\randoma'_{k\boldsymbol{\bigcdot}}$ for all $k\neq i,j$.
     \end{enumerate}

A probabilistic assignment is said to be SD-pair efficient if it is not SD-pair dominated by any other probabilistic assignment. A probabilistic assignment rule $\rrule$ is \textbf{SD-pair efficient} if for all $P_N\in \mathcal{D}^n$, $\rrule(P_N)$ is SD-pair efficient at $P_N$.
\end{definition}

\begin{remark}\label{rem_1}
    As we are dealing with strict preferences in this paper, the notion of SD-pair domination becomes equivalent to the following. A probabilistic assignment \randoma is SD-pair dominated at a profile $P_N$ by another probabilistic assignment \randoma'  if there exists $(i,j)\in N^2$  such that
\begin{enumerate}
\item [(i)] $\randoma'_{l\boldsymbol{\bigcdot}} \succ_{P_l} \randoma_{l\boldsymbol{\bigcdot}}$  for all  $l\in \{i,j\}$, and
    \item [(ii)] $\randoma_{k\boldsymbol{\bigcdot}}=\randoma'_{k\boldsymbol{\bigcdot}}$ for all $k\neq i,j$. 
     \end{enumerate}
     \hfill $\square$
\end{remark}

\cite{ekici2024pair} shows that pair efficiency (SD-pair efficiency in the deterministic setting) is much weaker than Pareto efficiency by providing an example (see Example 1 in \cite{ekici2024pair}) of a preference profile with at least seven agents that has only one Pareto efficient outcome, but more than $2^n$ pair efficient outcomes. We strengthen the ground even further by providing an example of a preference profile with at least three agents (notice that for two agents, SD-Pareto efficiency and SD-pair efficiency are equivalent) that has a unique SD-Pareto efficient outcome but infinitely many SD-pair efficient outcomes. 

\begin{example}\label{example_1}
    Assume $n>2$, and for simplicity we write $x_{n+s}=x_s$ for any $s\in \{1,\ldots,n\}$. Consider the following preference profile $P_N$.
\begin{center}
\[
\begin{tabular}{cccccc}
$P_1$ & $P_2$ & $P_3$ & $P_4$ & $\cdots$ & $P_n$ \\
\hline
$x_1$ & $x_2$ & $x_3$ & $x_4$ & $\cdots$ & $x_n$ \\
$x_2$ & $x_3$ & $x_4$ & $x_5$ & $\cdots$ & $x_1$ \\
$\vdots$ & $\vdots$ & $\vdots$ & $\vdots$ & $\cdots$ & $\vdots$ \\
\end{tabular}
\]
\end{center}

Now, for $b\in [0,1]$, consider the following probabilistic assignment $A^b$, where for all $i\in N$ 
\begin{align*}
     A^b_{ix_i}=b \text{ and } A^b_{ix_{i+1}}=1-b.
\end{align*}



Note that if $b<1$, $A^b$ is SD-Pareto dominated by $A^1$. Further, $A^1$ is the only allocation that is SD-Pareto efficient at $P_N$. However, each  $A^b$ is indeed SD-pair efficient at $P_N$. To see this, assume for contradiction that for some $b\in [0,1]$, $A^b$ is SD-pair dominated by another assignment $B$. This means there exists a pair of agents $(r,s)$ such that $A^b_{i\bigcdot}=B_{i\bigcdot}$ for all $i\notin \{r,s\}$ and $B_{i\bigcdot}\succ_{P_i}A^b_{i\bigcdot}$ for all $i\in \{r,s\}$. This, together with Remark \ref{rem_1} and the fact that $A^b$ assigns probabilities only to the top two objects of $P_r$ and $P_s$, it must be that under $B$, the probabilities of the top object of $P_r$ and $P_s$ is strictly more than that under $A^b$, i.e., $B_{rx_r}>b$ and $B_{sx_s}>b$. However, this implies that the assignments of agents $r-1$ and $s-1$ in $B$ must also differ from those in $A^b$ to maintain the bi-stochastic structure. Since $A^b_{i\bigcdot} = B_{i\bigcdot}$ for all $i \notin {r, s}$ and $n > 2$, we have a contradiction. 

This shows that there are infinitely many SD-pair efficient allocations (one for each value of $b$) at the profile $P_N$. 
\hfill $\square$
\end{example}





The following domain condition was also introduced in \cite{sen2011gibbard}. It is stronger than the FPT condition defined in the previous subsection.  
\begin{definition}
    A domain $\mathcal{D}$ is a Free Triple at the Top (FTT) domain if for all distinct $x,y,z\in X$ there exists $P\in \mathcal{D}$ such that $P\equiv xyz\cdots$.
\end{definition}

Our next result shows that on an FTT domain, a PAR satisfies SD-pair efficiency, SD-individual rationality, and SD-strategy-proofness if and only if it is the TTC rule.

\begin{theorem}\label{thm_2} Let $\mathcal{D}$ be a FTT domain. Then, a PAR  on $\mathcal{D}$ is  SD-pair-efficient, SD-individually rational, and SD-top-strategy-proof if and only if it is the TTC rule. 
\end{theorem}

\begin{remark}\label{rem_2}
    As SD-Pareto efficiency is stronger than SD-pair efficiency, one might think that Theorem \ref{thm_1} is a corollary of Theorem \ref{thm_2}. But this is not true, as Theorem \ref{thm_1} holds for FPT domains, whereas Theorem \ref{thm_2} requires an FTT domain, a stronger domain condition. 
\end{remark}

\cite{ekici2024pair} shows that, on the full preference domain $\mathcal{P}$, the TTC rule is the unique deterministic rule that satisfies pairwise efficiency, individual rationality, and strategy-proofness. We obtain this result as a corollary of Theorem~\ref{thm_2}.

\section{Results under Ex-post Notions}\label{sec: exp}
It is well-known that every bi-stochastic matrix can be written as a convex combination of permutation matrices (\cite{birkhoff1946three,von1950certain,budish2013designing}). In view of this, there is another way of defining Pareto-efficiency and Individual rationality in the probabilistic setting as introduced in  \cite{abdulkadirouglu1999house}.

\subsection{Results under Ex-post Pareto Efficiency}

We first define ex post Pareto efficiency and ex post individual rationality formally.

\begin{definition}(Ex post Pareto efficiency)\label{def: epe} A probabilistic assignment \randoma is Ex post Pareto efficient at a profile $P_N$ if it can be written as a convex combination of deterministic Pareto efficient assignments at $P_N$.
 
A PAR $\varphi:\mathcal{D}^n \to \mathcal{A}$ is \textit{ex post Pareto efficient} if for all $P_N\in \mathcal{D}^n$, $\varphi(P_N)$ is ex post Pareto efficient.	
\end{definition}

\begin{definition}(Ex post individual rationality)\label{def: eir} A PAR $\varphi:\mathcal{D}^n \to \mathcal{A}$ is \textit{ex post individually rational} (ex post IR) if for all $P_N\in \mathcal{D}^n$, $\varphi(P_N)$ is a convex combination of deterministic individually rational assignments in $P_N$.
\end{definition}

It is quite straightforward to see that SD-individual rationality is equivalent to ex post individual rationality when the initial endowment is deterministic.\footnote{{We provide a proof of this statement in Appendix \ref{appen_C} for completeness.}} However, SD-Pareto efficiency and ex post Pareto efficiency are not equivalent, in general.\footnote{{\cite{cho2016equivalence} provide a condition on a preference profile that is necessary and sufficient for the two notions to be equivalent.}} {It is known that an ex post Pareto efficient probabilistic assignment may fail to be SD-Pareto efficiency when $n\geq 4$, but every SD-Pareto probabilistic efficient assignment is ex post Pareto efficient.\footnote{{This is proved in \cite{bogomolnaia2001new}, Lemma 2-(ii).}} Below, we present an example in which a probabilistic outcome at a given preference profile is ex post efficient Pareto efficient, i.e, can be expressed as a convex combination of deterministic Pareto outcomes at the same profile. However, the probabilistic outcome is SD-Pareto dominated by another outcome at that profile, and hence, it is not SD-Pareto efficient.}\footnote{Similar examples can be found in \cite{bogomolnaia2001new} and \cite{abdulkadirouglu2003ordinal} (see Example 2).}  

\begin{example}\label{example_2}
   Suppose there are four agents, say $\{1,2,3,4\}$, and four alternatives, say $\{a,b,c,d\}$. Consider the preference profile $(P_1,P_2,P_3,P_4)$ shown in Table \ref{table_1}.
	\begin{table}
\begin{center}
\begin{tabular}{||c c c c||} 
		\hline
		$P_1$ & $P_2$ & $P_3$ & $P_4$ \\ [0.5ex] 
		\hline\hline
		$c$ & $a$ & {$a$} & {$c$} \\ 
		\hline
		{$a$} & {$c$} & {$b$} & {$d$} \\ 
		\hline
		{$b$} & {$d$} & $c$ & $a$ \\ 
		\hline
		$d$ & $b$ & $d$ & $b$ \\ 
		\hline
	\end{tabular}
    \caption{Table 1}\label{table_1}
\end{center}
\end{table}
Now, consider the probabilistic assignment $A$ given in the following table, where the rows represent the agents and the columns represent the objects, and the cells have the assignment probabilities.  
\[A=
\begin{blockarray}{ccccc}
	& a & b & c & d  \\
	\begin{block}{c[cccc]}
		1 & 	\tfrac{1}{2} & \tfrac{1}{2} & 0 & 0 \\
		2 & 0 & 0 & \tfrac{1}{2} & \tfrac{1}{2} \\
		3 & \tfrac{1}{2} & \tfrac{1}{2} & 0 & 0 \\
		4 & 0 & 0 & \tfrac{1}{2} & \tfrac{1}{2} \\
			\end{block}
\end{blockarray}
\]

First, observe that $A$ is SD-Pareto dominated by $B$ (given in Table \ref{table_2}), where the changes are marked in red. Thus, $A$ is not SD-Pareto efficient. 

\[B=
\begin{blockarray}{ccccc}
	& a & b & c & d  \\
	\begin{block}{c[cccc]}
		1 & \textcolor{red}{0} & \tfrac{1}{2} & 	\textcolor{red}{\tfrac{1}{2}} & 0 \\
		2 & 	\textcolor{red}{\tfrac{1}{2}} & 0 & \textcolor{red}{0} & \tfrac{1}{2} \\
		3 & \tfrac{1}{2} & \tfrac{1}{2} & 0 & 0 \\
		4 & 0 & 0 & \tfrac{1}{2} & \tfrac{1}{2} \\
	\end{block}
\end{blockarray}\label{table_2}
\]

However, in the following, we show that $A$ is ex post Pareto efficient by showing that $A$ can be decomposed into two deterministic assignments where both are Pareto-efficient at $(P_1,P_2,P_3,P_4)$. Consider the two deterministic assignments, $C$ and $D$, given in Table \ref{table_3}.

\begin{minipage}[t]{0.45\textwidth}
\centering
C\;=\begin{blockarray}{ccccc}
  & a & b & c & d \\
  \begin{block}{r[cccc]}
   1 & 1 & 0 & 0 & 0 \\
		2 &  0 & 0 & 0 & 1 \\
		3 & 0 & 1 & 0 & 0 \\
		4 & 0 & 0 & 1 & 0 \\
  \end{block}
\end{blockarray}
\end{minipage}%
\hfill 
\begin{minipage}[t]{0.45\textwidth}
\centering
D\;=\begin{blockarray}{ccccc}
  & a & b & c & d \\
  \begin{block}{r[cccc]}
   1 & 0 & 1 & 0 & 0 \\
		2 &  0 & 0 & 1 & 0 \\
		3 & 1 & 0 & 0 & 0 \\
		4 & 0 & 0 & 0 & 1 \\
  \end{block}
\end{blockarray}\label{table_3}
\end{minipage}

It can be seen that both $C$ and $D$ are Pareto-efficient at $(P_1,P_2,P_3,P_4)$ as $C$ is the outcome of the TTC rule at $(P_1,P_2,P_3,P_4)$ with the initial endowment $(a,d,b,c)$ and  $D$ is the outcome of the TTC rule at $(P_1,P_2,P_3,P_4)$ with the initial endowment $(b,c,a,d)$.\footnote{Recall that the initial endowment $(a,d,b,c)$ means agent $1$ is endowed with $a$, agent $2$ is endowed with $d$, and so on.} Thus, $A$ is ex post Pareto efficient. \hfill $\square$
\end{example}

Nevertheless, Theorem \ref{thm_3} shows that even under ex post Pareto efficiency and ex post individual rationality, the conclusion of Theorem \ref{thm_1} still holds.

\begin{theorem}\label{thm_3}
    Let $\mathcal{D}$ be an FPT domain. Then, a PAR on $\mathcal{D}$ is ex post Pareto efficient, ex post IR, and SD-top-strategy-proof if and only if it is the TTC rule. 
\end{theorem}

\begin{remark}\label{rem_3}
    As every SD-Pareto efficient assignment is also ex post Pareto efficient (Lemma 2-(ii) in \cite{bogomolnaia2001new}), we may see Theorem \ref{thm_1} as a corollary of Theorem \ref{thm_3}. 
\end{remark}

\subsection{Results under Ex-post pair Efficiency} 

In the spirit of ex post Pareto efficiency, we may similarly define ex post pair-efficiency. Below, we formally define it.

\begin{definition}(Ex post pair-efficiency)\label{def: epae} A probabilistic assignment \randoma is Ex post pair-efficient at a profile $P_N$ if it can be written as a convex combination of deterministic pair-efficient assignments at $P_N$.
 
A PAR $\varphi:\mathcal{D}^n \to \mathcal{A}$ is \textit{ex post pair-efficient} if for all $P_N\in \mathcal{D}^n$, $\varphi(P_N)$ is ex post pair-efficient.	
\end{definition}

{It can be observed that the probabilistic outcome in Example \ref{example_2} is actually ex post pair-efficient but not SD-pair efficient, as it can be improved for the agent pair $(1,2)$ without changing the assignments of the others. Therefore, an ex post pair efficient allocation may fail to satisfy SD-pair efficiency. However, similar to SD-Pareto efficiency, it can be shown that every SD-pair efficient allocation is ex post pair efficient. We omit the proof as it uses the same idea as in Lemma 2-(ii) in \cite{bogomolnaia2001new}.
 Nevertheless, like before, if we replace SD-pair efficiency by ex post pair efficiency, and SD-individual rationality by ex post individual rationality, the conclusion of Theorem \ref{thm_2} continues to hold.}

\begin{theorem}\label{thm_4}
    Let $\mathcal{D}$ be an FTT domain. Then, a PAR on $\mathcal{D}$ is ex post pair efficient, ex post IR, and SD-top-strategy-proof if and only if it is the TTC rule. 
\end{theorem}

\begin{remark}
    {In the same spirit of Remark \ref{rem_3}, we may see Theorem \ref{thm_2} as a corollary of Theorem \ref{thm_4}.}
\end{remark}

\section{Future Directions}\label{sec: fd}
Throughout the paper, we have considered deterministic endowments. A natural question arises: what happens when the initial endowment is probabilistic? It is shown in the proof of Theorem~3 of \cite{athanassoglou2011house} that when there are four agents and the initial endowment is uniform (that is, each element of the endowment, viewed as an assignment represented by a bi-stochastic matrix, is equal to \(\tfrac{1}{4} \)), there exists no probabilistic rule satisfying SD-Pareto efficiency, SD-individual rationality, and SD-strategy-proofness. This raises the question of whether this impossibility extends to arbitrary strictly probabilistic endowments. If not, what are the necessary and sufficient structures of the initial probabilistic endowments that allow for such a rule, or whether relaxing Pareto efficiency to the weaker notion of pair efficiency leads to a possibility? We leave these questions for future research.

	\bibliographystyle{plainnat}
	\setcitestyle{numbers}
	\bibliography{references.bib}

\begin{thebibliography}{22}
\providecommand{\natexlab}[1]{#1}
\providecommand{\url}[1]{\texttt{#1}}
\expandafter\ifx\csname urlstyle\endcsname\relax
  \providecommand{\doi}[1]{doi: #1}\else
  \providecommand{\doi}{doi: \begingroup \urlstyle{rm}\Url}\fi

\bibitem[Abdulkadiro{\u{g}}lu and S{\"o}nmez(1999)]{abdulkadirouglu1999house}
Atila Abdulkadiro{\u{g}}lu and Tayfun S{\"o}nmez.
\newblock House allocation with existing tenants.
\newblock \emph{Journal of Economic Theory}, 88\penalty0 (2):\penalty0
  233--260, 1999.

\bibitem[Abdulkadiro{\u{g}}lu and S{\"o}nmez(2003)]{abdulkadirouglu2003ordinal}
Atila Abdulkadiro{\u{g}}lu and Tayfun S{\"o}nmez.
\newblock Ordinal efficiency and dominated sets of assignments.
\newblock \emph{Journal of Economic Theory}, 112\penalty0 (1):\penalty0
  157--172, 2003.

\bibitem[Anno(2015)]{anno2015short}
Hidekazu Anno.
\newblock A short proof for the characterization of the core in housing
  markets.
\newblock \emph{Economics Letters}, 126:\penalty0 66--67, 2015.

\bibitem[Aswal et~al.(2003)Aswal, Chatterji, and Sen]{aswal2003dictatorial}
Navin Aswal, Shurojit Chatterji, and Arunava Sen.
\newblock Dictatorial domains.
\newblock \emph{Economic Theory}, 22\penalty0 (1):\penalty0 45--62, 2003.

\bibitem[Athanassoglou and Sethuraman(2011)]{athanassoglou2011house}
Stergios Athanassoglou and Jay Sethuraman.
\newblock House allocation with fractional endowments.
\newblock \emph{International Journal of Game Theory}, 40:\penalty0 481--513,
  2011.

\bibitem[Aziz(2015)]{aziz2015generalizing}
Haris Aziz.
\newblock Generalizing top trading cycles for housing markets with fractional
  endowments.
\newblock \emph{arXiv preprint arXiv:1509.03915}, 2015.

\bibitem[Birkhoff(1946)]{birkhoff1946three}
Garrett Birkhoff.
\newblock Three observations on linear algebra.
\newblock \emph{Univ. Nac. Tacuman, Rev. Ser. A}, 5:\penalty0 147--151, 1946.

\bibitem[Bogomolnaia and Heo(2012)]{bogomolnaia2012probabilistic}
Anna Bogomolnaia and Eun~Jeong Heo.
\newblock Probabilistic assignment of objects: Characterizing the serial rule.
\newblock \emph{Journal of Economic Theory}, 147\penalty0 (5):\penalty0
  2072--2082, 2012.

\bibitem[Bogomolnaia and Moulin(2001)]{bogomolnaia2001new}
Anna Bogomolnaia and Herv{\'e} Moulin.
\newblock A new solution to the random assignment problem.
\newblock \emph{Journal of Economic theory}, 100\penalty0 (2):\penalty0
  295--328, 2001.

\bibitem[Budish et~al.(2013)Budish, Che, Kojima, and
  Milgrom]{budish2013designing}
Eric Budish, Yeon-Koo Che, Fuhito Kojima, and Paul Milgrom.
\newblock Designing random allocation mechanisms: Theory and applications.
\newblock \emph{American economic review}, 103\penalty0 (2):\penalty0 585--623,
  2013.

\bibitem[Cho and Do{\u{g}}an(2016)]{cho2016equivalence}
Wonki~Jo Cho and Battal Do{\u{g}}an.
\newblock Equivalence of efficiency notions for ordinal assignment problems.
\newblock \emph{Economics Letters}, 146:\penalty0 8--12, 2016.

\bibitem[Chun and Yun(2020)]{chun2020upper}
Youngsub Chun and Kiyong Yun.
\newblock Upper-contour strategy-proofness in the probabilistic assignment
  problem.
\newblock \emph{Social Choice and Welfare}, 54:\penalty0 667--687, 2020.

\bibitem[Ekici(2024)]{ekici2024pair}
{\"O}zg{\"u}n Ekici.
\newblock Pair-efficient reallocation of indivisible objects.
\newblock \emph{Theoretical Economics}, 19\penalty0 (2):\penalty0 551--564,
  2024.

\bibitem[Gibbard(1977)]{gibbard1977manipulation}
Allan Gibbard.
\newblock Manipulation of schemes that mix voting with chance.
\newblock \emph{Econometrica: Journal of the Econometric Society}, pages
  665--681, 1977.

\bibitem[Gogulapati et~al.(2025)Gogulapati, Narahari, Roy, and
  Sadhukhan]{gogulapati2025probabilisticassignmentrules}
Sreedurga Gogulapati, Yadati Narahari, Souvik Roy, and Soumyarup Sadhukhan.
\newblock On probabilistic assignment rules, 2025.
\newblock URL \url{https://arxiv.org/abs/2507.09550}.

\bibitem[Ma(1994)]{ma1994strategy}
Jinpeng Ma.
\newblock Strategy-proofness and the strict core in a market with
  indivisibilities.
\newblock \emph{International Journal of Game Theory}, 23\penalty0
  (1):\penalty0 75--83, 1994.

\bibitem[Sen(2011)]{sen2011gibbard}
Arunava Sen.
\newblock The gibbard random dictatorship theorem: a generalization and a new
  proof.
\newblock \emph{SERIEs}, 2\penalty0 (4):\penalty0 515--527, 2011.

\bibitem[Sethuraman(2016)]{sethuraman2016alternative}
Jay Sethuraman.
\newblock An alternative proof of a characterization of the ttc mechanism.
\newblock \emph{Operations Research Letters}, 44\penalty0 (1):\penalty0
  107--108, 2016.

\bibitem[Shapley and Scarf(1974)]{shapley1974cores}
Lloyd Shapley and Herbert Scarf.
\newblock On cores and indivisibility.
\newblock \emph{Journal of mathematical economics}, 1\penalty0 (1):\penalty0
  23--37, 1974.

\bibitem[Svensson(1994)]{svensson1994queue}
Lars-Gunnar Svensson.
\newblock Queue allocation of indivisible goods.
\newblock \emph{Social Choice and Welfare}, 11\penalty0 (4):\penalty0 323--330,
  1994.

\bibitem[Von~Neumann(1950)]{von1950certain}
John Von~Neumann.
\newblock A certain zero-sum two-person game equivalent to the optimal
  assignment problem$^1$.
\newblock \emph{Contributions to the Theory of Games}, \penalty0 (24):\penalty0
  5, 1950.

\bibitem[Y{\i}lmaz(2009)]{yilmaz2009random}
{\"O}zg{\"u}r Y{\i}lmaz.
\newblock Random assignment under weak preferences.
\newblock \emph{Games and Economic Behavior}, 66\penalty0 (1):\penalty0
  546--558, 2009.

\end{thebibliography}

\appendix
\section{Proof of Theorem \ref{thm_1}}
\begin{proof}
The if part of the result follows from \cite{ma1994strategy} as the TTC rule satisfies the three properties in the deterministic setup. We proceed to show the only-if part by induction on $n$. As the base case, first assume that $n=2$. If the top-ranked objects of each of the two agents are different, both of them get their top-ranked objects, and the outcome is that of the TTC rule. If the top-ranked objects of both the agents are the same, i.e., $\dendow{1}$, by SD-individual rationality, agent $1$ is assigned the whole of $\dendow{1}$ and hence agent 2 is assigned the whole of $\dendow{2}$. This again is the outcome of the TTC rule.

\textit{Induction Hypothesis (IH)}: When there are at most $n-1$ agents and objects, any assignment rule $\rrule$ satisfying SD-Pareto-efficiency, SD-individual rationality, and SD-top-strategy-proofness is the TTC rule.

We use IH to prove that the same statement holds for $n$ agents and objects. Let $P_N$ be the preference profile of the agents and construct $\mathcal{G}(P_N)$ such that an edge $(i,j)$ is present if and only if $\prefof{}{}{} = \dendow{j}$. WLG let $\mathcal{C}=(1,\ldots,k,1)$ be a cycle in $\mathcal{G}(P_N)$ (we use ${k+1}$ to denote $1$ and $0$ to denote $k$). First assume that  $\mathcal{C}$ is a singleton cycle, i.e., $k=1$ and hence, $P_{1}(1)=x_{1}$. Consequently, by IR,  $\rruleof{1}{\dendow{1}}{P_N} = 1$. Now, there are $n-1$ remaining agents and objects, and these remaining objects are the initial endowments of these agents. Therefore, by IH, the objects will be allocated according to the TTC rule, implying $\varphi(P_N)$ is the TTC outcome and completing the proof.

Now assume that $\mathcal{C}$ is a non-singleton cycle, i.e., $k\geq 2$. Consider a new profile $P'_N$  by modifying the preferences of agents as follows:
$$
P_i'\equiv\begin{cases}
			x_{i+1}x_i\cdots, & \text{if $i\in C$}\\
            P_i, & \text{otherwise.}
		 \end{cases}
$$

Two observations to be noted here. First, the preference profile $P_N'$ is feasible as $\mathcal{D}$ is an FPT domain, second, as $\mathcal{G}(P_N) = \mathcal{G}(P'_N)$, hence $\mathcal{C}$ remains a non-singleton cycle at $\mathcal{G}(P'_N)$ as well. We first prove the following claim, which states that at $P'_N$, all the objects owned by the agents in $\mathcal{C}$ must be assigned only to the agents in $\mathcal{C}$ with probability 1.
\begin{claim}\label{cla: ir}
For all $j \in \{1,\ldots,k\}$, we have,
$$\rruleof{{j-1}}{\dendow{j}}{P'_N}+\rruleof{j}{\dendow{j}}{P'_N} = 1.$$
\end{claim}
\begin{proof}
As $\rrule$ is SD-individually rational, for every $i \in \mathcal{C}$ we have, $\rruleof{i}{\{x_{i+1},\dendow{i}\}}{P'_N} = 1$. This, together with  $|\mathcal{C}|=k$, yields 
\begin{align}
    \label{eq: equalsz}
    \suml{i \in \mathcal{C}}{(\rrule_{ix_{i+1}}(P_N')+\rrule_{i\dendow{i}}(P_N'))}=k
\end{align}
Further, as $\mathcal{C}$ is a cycle, (\ref{eq: equalsz}) can be rewritten as
\begin{align}
    \label{eq: equalsz2}
    \suml{\{\dendow{j}\mid \; j \in \mathcal{C}\}} {(\rruleof{{j-1}}{\dendow{j}}{P'_N}+\rruleof{j}{\dendow{j}}{P'_N})} = k.
\end{align}
Since $\rruleof{\boldsymbol{\bigcdot}}{\dendow{j}}{P'_N}$ is a probability distribution, each term on the left side of (\ref{eq: equalsz2}) must be at most $1$. That is, for any $\dendow{j}$ such that $i_j \in \mathcal{C}$, we have
$\rruleof{{j-1}}{\dendow{j}}{P'_N}+\rruleof{j}{\dendow{j}}{P'_N} \leq 1$. This, when combined with (\ref{eq: equalsz2}) and the fact that $|\mathcal{C}| = k$, gives
$$\rruleof{{j-1}}{\dendow{j}}{P'_N}+\rruleof{j}{\dendow{j}}{P'_N} = 1.$$
Hence, the claim follows.
\end{proof}
We now use SD-Pareto efficiency of $\rrule$. From Claim (\ref{cla: ir}), endowments of all the agents in $\mathcal{C}$ are assigned to the agents in $\mathcal{C}$ with probability 1. We claim that $\varphi(P'_N)$ assigns $x_{i+1}$ to each agent $i \in \mathcal{C}$. To see this, consider any other assignment $A$ satisfying Claim \ref{cla: ir}. Clearly, $A$ is Pareto-dominated by an assignment $A'$ where 
$$
\randoma'_{i\boldsymbol{\bigcdot}}=\begin{cases}
			x_{i+1}, & \text{if $i\in C$}\\
            \randoma_{i\boldsymbol{\bigcdot}}, & \text{otherwise.}
		 \end{cases}
$$
Thus, by the SD-Pareto efficiency of $\varphi$,  $\varphi(P'_N)$ assigns $\prefof{}{}{}$ to each agent $i \in \mathcal{C}$. Also, this is the TTC outcome of the agents in $\mathcal{C}$. Formally, we have the following claim.    
\begin{claim}\label{cla: pe}
$\varphi(P_N')$ is the TTC assignment for all the agents in the cycle $\mathcal{C}$.
\end{claim}
We now proceed to show that the same as in Claim \ref{cla: pe} holds for the profile $P_N$ as well. We  sequentially change the preferences of the agents in $\mathcal{C}$ from $P_i'$ to $P_i$ for $i\in \{1,\ldots,k\}$. First, consider the case when only the preference of agent 1 is changed to $P_1$, i.e., the underlying profile is  $P^1_N$ where $P_1^1=P_1$ and $P_i^1=P_i'$ for all $i\in \{2,\ldots,k\}$. As $\rruleof{1}{x_{1}}{P'_N} = 1$,  by SD-top-strategy-proofness, we have $\rruleof{1}{x_{1}}{P^1_N} = 1$. Now for agent $2$, as $P_2'\equiv x_3\dendow{2}\cdots$, by SD-individual rationality, $\rruleof{2}{\{x_3,\dendow{2}\}}{P^1_N} = 1$. This, together with $\rruleof{1}{x_{1}}{P^1_N} = 1$, implies $\rruleof{2}{x_3}{P^1_N} = 1$. We can continue this way and show that   $\rruleof{i}{x_{i+1}}{P^1_N} = 1$ for all $i\in \{3,\ldots,k\}$. Thus, the outcome at $P_N^1$ is the TTC outcome for the agents in $\mathcal{C}$.

We now use an induction on the number of agents whose preferences are changed to the preferences in $P_N$. The base case is already shown in the previous paragraph.  Suppose the outcome is the TTC outcome for the agents in $\mathcal{C}$ after changing the preferences of agents $1,\ldots,l$, i.e., for the profile $P^l_N$ where 
$$
P^l_i=\begin{cases}
			P'_i, & \text{if $i\in \{l+1,\ldots,k\}$}\\
            P_i, & \text{otherwise.}
		 \end{cases}
$$
We show that the same holds for $P_N^{l+1}$. By SD-top-strategy-proofness, $\rruleof{l+1}{x_{l+2}}{P^{l+1}_N} = 1$. For agent $l+2$, as $P_{l+2}'\equiv x_{l+3}\dendow{l+2}\cdots$, by SD-individual rationality, $\rruleof{l+2}{\{x_{l+3},\dendow{l+2}\}}{P^{l+1}_N} = 1$. As $\rruleof{l+1}{x_{l+2}}{P^{l+1}_N} = 1$, this means $\rruleof{l+2}{x_{l+3}}{P^{l+2}_N} = 1$. Continuing in this way we have $\rruleof{i}{x_{i+1}}{P^{l+1}_N} = 1$ for all $i\in \{l+3,\ldots,k\}$. We are now left with showing that $\rruleof{i}{x_{i+1}}{P^{l+1}_N} = 1$ for all $i\in \{1,\ldots,l\}$. We show it for agent 1, and the similar arguments hold for other agents too. Consider the preference $\bar{P}_1\equiv x_2x_{1}\cdots$. Note that in the proof so far, we have not used any restriction on $P_1$ other than the top-ranked alternative of $P_1$ is $x_2$. Thus, the same arguments can be carried out to show that at $(\bar{P}_1,P^{l+1}_{-1})$, $\rruleof{i}{x_{i+1}}{\bar{P}_1,P^{l+1}_{-1}}=1$ for all $i\in \{l+1,\ldots,k\}$. This, in particular, $\rruleof{k}{x_1}{\bar{P}_1,P^{l+1}_{-1}}=1$,  together with SD-individual rationality for agent 1, yields $\rruleof{1}{x_2}{\bar{P}_1,P^{l+1}_{-1}}=1$. Now, by SD-top-strategy-proofness, $\rruleof{1}{x_2}{P^{l+1}_N}=1$. This completes the induction step, and hence, it is proved that the outcome at $P_N$ coincides with the TTC outcome. The remaining proof follows by IH as there are fewer than $n-1$ agents and objects; these remaining objects are the initial endowments of these agents.
\end{proof}

\section{Proof of Theorem \ref{thm_2}}

\begin{proof}
    From \cite{ekici2024pair}'s result, it follows that the TTC rule is SD-pair-efficient, SD-individually rational and SD-top-strategy-proof. Now we prove the only-if direction using induction on $n$. First, assume that $n = 2$. As for $n=2$, SD-Pareto efficiency and SD-pair efficiency are identical; it follows from Theorem \ref{thm_1} that any PAR satisfying SD-pair efficiency, SD-individual rationality, and SD-top-strategy-proofness must be the TTC rule. We now consider the following Induction Hypothesis.
    

   \noindent \textit{Induction Hypothesis (IH)}: When there are at most $n-1$ agents and objects, every PAR $\rrule$ satisfying SD-pair-efficiency, SD-individual rationality, and SD-top-strategy-proofness is the TTC rule.

    We use IH to prove that the same statement holds for $n$ agents and objects. We start with a lemma that will be used quite a few times in the proof. Recall that our assumption the initial endowment for agent $i$ is $x_i$ for all $i\in \{1,\ldots,n\}$. Consider $k$ such that $1\leq k\leq n$, and consider the following preferences of the agents $1,\ldots,k$.
   \begin{equation}\label{eq_0}
\bar{P}_i\equiv\begin{cases}
			x_2x_1\cdots, & \text{if $i = 1$}\\
            x_{i+1}x_1x_i\cdots, & \text{if $i\in \{2,\ldots,k-1\}$}\\
            x_1x_k\cdots, & \text{if $i = k$.}
		 \end{cases}
  \end{equation} 

Note that such preferences are feasible as $\mathcal{D}$ is an FTT domain. Let $l\in \{1,\ldots,k\}$ and we denote the set of agents $\{l,\ldots,k\}$ by $L$. Further, let $\bar{P}_L=(\bar{P}_l,\ldots,\bar{P}_k)$.  We now formally state the lemma. In the statement and in the proof of the lemma, for notational convenience, we sometimes use $k+1$ to denote $1$ and $0$ to denote $k$.
\begin{lemma}\label{lemma_1}
For all $P_{-L}\in \mathcal{D}^{n-k+l-1}$, the following implication holds
$$\big[\text{either } l=1 \text{ or } \varphi_{l\dendow{l}}(\bar{P}_L,P_{-L})=0\big] \implies \big[\varphi_{i\dendow{i+1}}(\bar{P}_L,P_{-L})=1 \text{ for all }i\in L\big].$$
\end{lemma}
\begin{proof}
    Fix $P_{-L}\in \mathcal{D}^{n-k+l-1}$. We complete the proof in two steps. In the first step, we prove the following claim. 
    \begin{claim}\label{cl_1} The following two statements hold
    \begin{enumerate}[(i)]
        \item $\suml{i \in \{l,\ldots,k\}} {\rruleof{i}{\dendow{1}}{\bar{P}_L,P_{-L}}}=1$,  and 
        \item for all $j\in \{l+1,\ldots,k\}$, 
 $\big[\rruleof{{j-1}}{\dendow{j}}{\bar{P}_L,P_{-L}}+\rruleof{j}{\dendow{j}}{\bar{P}_L,P_{-L}}\big]= 1.$
    \end{enumerate}
    \end{claim}
   \noindent \textbf{Proof of the claim:} As $\rrule$ is SD-individually rational, we have the following
   \begin{enumerate} [(a)]
       \item $\rruleof{k}{\{x_1,x_k\}}{\bar{P}_L,P_{-L}}= 1$, and
       \item for every $i \in \{l+1,\ldots,k-1\}$, $\rruleof{i}{\{x_{i+1},x_1,\dendow{i}\}}{\bar{P}_L,P_{-L}} = 1$.
   \end{enumerate}
   Further, as either $l=1$ or $\varphi_{l\dendow{l}}(\bar{P}_L,P_{-L})=0$, by SD-individual rationality of $\varphi$ for $l$, we have $\rruleof{l}{\{x_{l+1},x_1\}}{\bar{P}_L,P_{-L}}= 1$. Thus, we have 
\begin{align}\label{eq_0}
   \rruleof{l}{\{x_{l+1},x_1\}}{\bar{P}_L,P_{-L}}+\rruleof{k}{\{x_1,x_k\}}{\bar{P}_L,P_{-L}} + \suml{i \in \{2,\ldots,k-1\}}\rrule_{i\{x_{i+1},x_1,x_i\}}(\bar{P}_L,P_{-L})=k-l+1.
\end{align}
Rearranging the terms in (\ref{eq_0}), we may write
\begin{align*}
\suml{i \in \{l,\ldots,k\}} {\rruleof{i}{\dendow{1}}{\bar{P}_L,P_{-L}}}
+\suml{j\in \{l+1,\ldots,k\}} {\big(\rruleof{{j-1}}{\dendow{j}}{\bar{P}_L,P_{-L}}+\rruleof{j}{\dendow{j}}{\bar{P}_L,P_{-L}}\big)} = k-l+1.
\end{align*}
Since $\rruleof{\boldsymbol{\bigcdot}}{\dendow{i}}{P'_N}$ is a probability distribution for all $i\in \{l+1,\ldots,k\}\cup \{1\}$, we have
\begin{align*}
  \underbrace{\suml{i \in \{l,\ldots,k\}} {\rruleof{i}{\dendow{1}}{\bar{P}_L,P_{-L}}}}_{\leq 1}
+\underbrace{\suml{j\in \{l+1,\ldots,k\}} {\big[\underbrace{\rruleof{{j-1}}{\dendow{j}}{\bar{P}_L,P_{-L}}+\rruleof{j}{\dendow{j}}{\bar{P}_L,P_{-L}}}_{\leq 1}\big]}}_{\leq k-l-2}=k-l-1.
\end{align*}
Hence, we have 
\begin{align*}
&\suml{i \in \{l,\ldots,k\}} {\rruleof{i}{\dendow{1}}{\bar{P}_L,P_{-L}}}=1, \text{ and } 
 \big[\rruleof{{j-1}}{\dendow{j}}{\bar{P}_L,P_{-L}}+\rruleof{j}{\dendow{j}}{\bar{P}_L,P_{-L}}\big]=1 \text{ for all } j\in \{l+1,\ldots,k\}.
\end{align*}
This completes the proof of the claim. \hfill $\square$

We now proceed to complete the proof of the lemma. In order to do so, we prove another claim. 
\begin{claim}\label{cl_2}
    $\rruleof{{k}}{\dendow{1}}{\bar{P}_L,P_{-L}}=1$.
\end{claim}
\noindent \textbf{Proof of the claim:} Note that if $l=k$ then by Claim \ref{cl_1}-(i), there is nothing to prove. So assume that $l\leq k-1$ and
assume for contradiction that $\rruleof{{k}}{\dendow{1}}{\bar{P}_L,P_{-L}}<1$, which in view of Claim \ref{cl_1} implies that for some $i\in \{l,\ldots,k-1\}$, $\rruleof{{i}}{\dendow{1}}{\bar{P}_L,P_{-L}}>0$. We contradict this by constructing a probabilistic assignment $A$ (a bi-stochastic matrix) that SD-pair-dominates $\varphi(\bar{P}_L,P_{-L})$ for the agents $(i,i+1)$. Construct $A$ as follows:
\begin{enumerate}[(i)]
    \item $\randoma_{ix_{i+1}}=\rruleof{{i}}{\dendow{i+1}}{\bar{P}_L,P_{-L}}+\rruleof{{i}}{\dendow{1}}{\bar{P}_L,P_{-L}}$, $\randoma_{ix_{1}}=0$,  and  $\randoma_{ix_{i}}=\rruleof{{i}}{\dendow{i}}{\bar{P}_L,P_{-L}}$,
    \item  $\randoma_{i+1x_{i+2}}=\rruleof{{i+1}}{\dendow{i+2}}{\bar{P}_L,P_{-L}}$, $\randoma_{i+1x_{1}}=\rruleof{{i+1}}{\dendow{1}}{\bar{P}_L,P_{-L}}+\rruleof{{i}}{\dendow{1}}{\bar{P}_L,P_{-L}}$, and $\randoma_{i+1x_{i+1}}=\rruleof{{i+1}}{\dendow{i+1}}{\bar{P}_L,P_{-L}}-\rruleof{{i}}{\dendow{1}}{\bar{P}_L,P_{-L}}$, and
    \item all the rows of $A$, except rows $i$ and $i+1$, are the same as $\varphi(P_N')$, i.e., 
$\randoma_{r\boldsymbol{\bigcdot}}=\rruleof{r}{}{\bar{P}_L,P_{-L}}$ for all $r\in N\setminus \{i,i+1\}$. 
\end{enumerate}
We first verify that $A$ is a valid bi-stochastic matrix, implying (i) $\randoma_{rx_s}\geq 0$ for all $r,s\in \{1,\ldots,n\}$ and (ii) $\suml{r}\randoma_{rx_s}=\suml{s}\randoma_{rx_s}=1$ for all $r,s\in \{1,\ldots,n\}$. Now, (ii) follows as $\suml{r}A_{rx_s}=\suml{r}\rruleof{{r}}{\dendow{s}}{\bar{P}_L,P_{-L}}$ and $\suml{s}A_{rx_s}=\suml{r}\rruleof{{r}}{\dendow{s}}{\bar{P}_L,P_{-L}}$ for all $r,s\in \{1,\ldots,n\}$. To see (i), we only need to show $\randoma_{i+1x_{i+1}}\geq 0$ as all the other terms are either the same as that of $\varphi(\bar{P}_L,P_{-L})$ or sum of two elements of $\varphi(\bar{P}_L,P_{-L})$. Note that 
\begin{align*}
    \randoma_{i+1x_{i+1}}&=\rruleof{{i+1}}{\dendow{i+1}}{\bar{P}_L,P_{-L}}-\rruleof{{i}}{\dendow{1}}{\bar{P}_L,P_{-L}} \\
    &\geq  \rruleof{{i+1}}{\dendow{i+1}}{\bar{P}_L,P_{-L}}-\rruleof{{i}}{\dendow{1}}{\bar{P}_L,P_{-L}}-\rruleof{{i}}{\dendow{i}}{\bar{P}_L,P_{-L}}\\
    &=\rruleof{{i+1}}{\dendow{i+1}}{\bar{P}_L,P_{-L}}-(1-\rruleof{{i}}{\dendow{i+1}}{\bar{P}_L,P_{-L}})\\
    & =0.
\end{align*}
Furthermore, as $\rruleof{{i}}{\dendow{1}}{\bar{P}_L,P_{-L}}>0$, we have $\randoma_{ix_{i+1}}>\rruleof{{i}}{\dendow{i+1}}{\bar{P}_L,P_{-L}}$, implying $A\neq \varphi(\bar{P}_L,P_{-L})$. 

Observe that as $\bar{P}_{i}\equiv x_{i+1}x_1x_i\cdots$ and $\bar{P}_{i+1}\equiv x_{i+2}x_1x_{i+1}\cdots$, $A$ SD-pair-dominates $\varphi(\bar{P}_L,P_{-L})$ for agents pair $(i,i+1)$, a contradiction. Therefore, $\rruleof{{k}}{\dendow{1}}{\bar{P}_L,P_{-L}}=1$, and the proof of Claim \ref{cl_2} is complete. \hfill $\square$

What remains to complete the proof of Lemma \ref{lemma_1} is to show that $\rruleof{{j}}{\dendow{j+1}}{\bar{P}_L,P_{-L}}=1$ for all $l\leq j\leq k-1$. Again, if $l=k$, there is nothing to prove. So, we assume $l\leq k-1$. Note that by SD-individual rationality of $\varphi$ for agent $l$, we have $\varphi_{1\{\dendow{2},\dendow{1}\}}(\bar{P}_L,P_{-L})=1$ if $l=1$ and $\varphi_{l\{\dendow{l+1},\dendow{1},\dendow{l}\}}(\bar{P}_L,P_{-L})=1$ if $l>1$. This, together with Claim \ref{cl_2} and the assumption that either $l=1$ or $\varphi_{l\dendow{l}}(\bar{P}_L,P_{-L})=0$, implies $\varphi_{l\dendow{l+1}}(\bar{P}_L,P_{-L})=1$. If $l+1=k-1$, there is nothing to prove further. If $l+1<k-1$, then by SD-individual rationality of $\varphi$ for $l+1$, we have $\varphi_{l+1\{\dendow{l+2},\dendow{1},\dendow{l+1}\}}(\bar{P}_L,P_{-L})=1$. As we have already shown $\varphi_{l\dendow{l+1}}(\bar{P}_L,P_{-L})=1$ and by Claim \ref{cl_2}, $\varphi_{l+1\dendow{1}}(\bar{P}_L,P_{-L})=0$, it follows that $\varphi_{l+1\dendow{l+2}}(\bar{P}_L,P_{-L})=1$. Continuing in this way, we can show that $\rruleof{{j}}{\dendow{j+1}}{\bar{P}_L,P_{-L}}=1$ for all $j\in \{l,\ldots,k-1\}$. This completes the proof of the lemma. \end{proof}



   We are now ready to prove the theorem. Let $P_N\in \mathcal{D}^n$ be a preference profile of the agents, and we construct $\mathcal{G}(P_N)$ as follows: an edge $(i,j)$ is present if and only if $\prefof{}{}{} = \dendow{j}$. WLG let $\mathcal{C}=(1,\ldots,p,1)$ be a cycle in $\mathcal{G}(P_N)$ (we use ${p+1}$ to denote $1$ and $0$ to denote $k$), i.e, $\prefof{}{}{} = \dendow{i+1}$ for all $i\in \{1,\ldots,p\}$. First assume that  $\mathcal{C}$ is a singleton cycle, i.e., $p=1$ and hence, $P_{1}(1)=x_{1}$. Consequently, by IR,  $\rruleof{1}{\dendow{1}}{P_N} = 1$. Now, there are $n-1$ remaining agents and objects, and these remaining objects are the initial endowments of these agents. Therefore, by IH, the objects will be allocated according to the TTC rule, implying $\varphi(P_N)$ is the TTC outcome and completing the proof.

    Now assume that $\mathcal{C}$ is a non-singleton cycle, i.e., $k\geq 2$. Consider the preference profile $\tilde{P}_N$:
    $$
\tilde{P}_i\equiv\begin{cases}
			\bar{P}_i, & \text{if $i \in \{1,\ldots,p\}$}\\
            P_i, & \text{otherwise,}
		 \end{cases}
$$
where $\bar{P}_i$ is defined in (\ref{eq_0}) with $k=p$. In view of Lemma \ref{lemma_1} with $l=1$ and $k=p$, we can claim that $\varphi_{i\dendow{i+1}}(\bar{P}_L,P_{-L})=1$  for all $i\in \{1,\ldots,k\}$. Thus, the outcome of $\varphi$ at 
$\tilde{P}_N$ for the agents in the cycle $\mathcal{C}$ is the TTC outcome. We claim that for the rest of the agents, also, the outcome is the TTC outcome. This follows from IH as there are fewer than $n-1$ agents and objects remaining, with these remaining objects being the initial endowments of these remaining agents.


We now proceed to show that the outcome at $P_N$ is also the TTC outcome. Here we invoke SD-top-strategy-proofness. We  sequentially change the preferences of the agents in $\mathcal{C}$ from $\tilde{P}_i$ to $P_i$ for $i\in \{1,\ldots,k\}$. First, consider the case when only the preference of agent 1 is changed to $P_1$, i.e., the underlying profile is $P^1_N$
$$
P^1_i\equiv\begin{cases}
			\tilde{P}_i, & \text{if $i \in \{2,\ldots,p\}$}\\
            P_i, & \text{otherwise,}
		 \end{cases}
$$
We show that $\varphi(P^1_N)=\varphi(\tilde{P}_N)$. As $\rruleof{1}{\dendow{2}}{\tilde{P}_N} = 1$ and $\tilde{P}_1(1)=P^1_1(1)=x_2$, by SD-top-strategy-proofness, we have $\rruleof{1}{\dendow{2}}{P^1_N} = 1$. For other agents in $\mathcal{C}$, we use Lemma \ref{lemma_1} with $l=2$ and $k=p$. The condition that $\rruleof{2}{\dendow{2}}{P^1_N} = 0$ follows as we have  $\rruleof{1}{\dendow{2}}{P^1_N} = 1$. Thus, by Lemma \ref{lemma_1}, $\rruleof{i}{\dendow{i+1}}{\tilde{P}_N}=1$ for all $i\in \{1,\ldots,p\}$; the outcome at $P_N^1$ is the TTC outcome for the agents in $\mathcal{C}$. Again, we can use IH and show that the conclusion holds for all the agents implying $\varphi(P^1_N)$ is the TTC outcome. 

We now use another induction on the number of agents whose preferences are changed to the preferences in $P_N$. The base case is already shown in the previous paragraph.  Suppose the outcome is the TTC outcome for the agents in $\mathcal{C}$ after changing the preferences of agents $1,\ldots,t$, i.e., for the profile $P^t_N$ where 
$$
P^t_i=\begin{cases}
			\tilde{P}_i, & \text{if $i\in \{t+1,\ldots,k\}$}\\
            P_i, & \text{otherwise.}
		 \end{cases}
$$
We show that the same holds for $P_N^{t+1}$. As $\tilde{P}_{t+1}(1)=P_{t+1}(1)=x_{t+2}$ and $\rruleof{t+1}{\dendow{t+2}{}{}}{P^{t}_N} = 1$, by SD-top-strategy-proofness, $\rruleof{t+1}{\dendow{t+2}{}{}}{P^{t+1}_N} = 1$. For agents in $\{t+2,\ldots,p\}$, we can use Lemma \ref{lemma_1} with $l=t+2$ and $k=p$ where the condition $\rruleof{t+2}{\dendow{t+2}{}{}}{P^{t+1}_N} = 0$ follows as $\rruleof{t+1}{\dendow{t+2}{}{}}{P^{t+1}_N} = 1$. Thus, we have $\rruleof{i}{\dendow{i+1}{}{}}{P^{t+1}_N} = 1$ for all $i\in \{t+2,\ldots,k\}$.

We are now left with showing that $\rruleof{i}{\dendow{i+1}{}{}}{P^{t+1}_N} = 1$ for all $i\in \{1,\ldots,t\}$. We show this for $i=1$, and the same arguments hold for other agents as well. Note that in the proof so far, we have not used any restriction on $P_1$ other than $P_1(1)=x_2$, implying the conclusions in the previous paragraph hold for a particular choice of $P_1$, say $\hat{P}_1\equiv x_2x_{t+2}\cdots$. Therefore, we have $\rruleof{t+1}{\dendow{t+2}{}{}}{\hat{P}_1,P^{t+1}_{-1}} = 1$. Further, by SD-individual rationality of $\varphi$ for agent 1, $\rruleof{1}{\{\dendow{2},{\dendow{t+2}\}}{}}{\hat{P}_1,P^{t+1}_{-1}} = 1$. These two observations together imply $\rruleof{1}{\dendow{2}{}{}}{\hat{P}_1,P^{t+1}_{-1}} = 1$. Now using SD-top-strategy-proofness for agent 1, we may conclude that $\rruleof{1}{\dendow{2}{}{}}{P^{t+1}_{N}} = 1$. This shows that $\rruleof{i}{\dendow{i+1}{}{}}{P^{t+1}_N} = 1$ for all $i\in \{1,\ldots,t\}$. Now the rest of the proof follows using the IH as the number of agents outside the cycle $\mathcal{C}$ is strictly less than $n-1$.
\end{proof}

\section{Equivalence of SD-individual rationality and ex post individual rationality}\label{appen_C}
\begin{lemma}\label{lem_1}
    A PAR $\varphi:\mathcal{D}^n \to \mathcal{A}$ is SD-individually rational if and only if it is ex post individually rational when the initial endowment is deterministic.
\end{lemma}
\begin{proof}
    (\textbf{If part:}) Recall that $E$ is the deterministic initial endowment, and let $\varphi$ be a PAR satisfying ex post individual rationality. Consider a profile $P_N\in \mathcal{D}^n$. As $\varphi$ satisfies ex post individual rationality, this means $\varphi(P_N)$ can be written as a convex combination of deterministic assignments $\varphi(P_N)=\sum_{t=1}^k\alpha_tf^t(P_N)$ with $\alpha_t\geq 0$ for all $t\leq k$ and $\sum_{t=1}^k\alpha_t=1$ such that  $f^t(P_N)$ is individually rational at $P_N$ for all $t\in\{1,\ldots,k\}$. This means for any $i\in N$, $f_i^t(P_N)R_i x_i$ for all $t\in N$ implying $\varphi_i(P_N)\succeq_{P_i}x_i$. Hence, $\varphi$ is SD-individually rational.

    (\textbf{Only-if part:}) Now suppose $\varphi$ is SD-individually rational where the initial endowment is deterministic. Consider a profile $P_N\in \mathcal{D}^n$. By SD-individual rationality $\varphi_i(P_N)\succeq_{P_i}x_i$ for all $i\in N$. Moreover, as $\varphi(P_N)$ can be written as a convex combination of deterministic assignments (permutation matrices), say $\varphi(P_N)=\sum_{t=1}^k\alpha_t Q^t$, with $\alpha_t\geq 0$ for all $t\leq k$ and $\sum_{t=1}^k\alpha_t=1$, it must be that $Q^t_{i\cdot}\succeq_{P_i}x_i$ for all $t\in \{1,\ldots,k\}$. Thus, $\varphi$ is ex-post individually rational.
 \end{proof}

\section{Proof of Theorem \ref{thm_3}}
\begin{proof}
    We go along the lines of the proof of Theorem \ref{thm_1}. As SD-individual rationality is equivalent to ex post individual rationality, we only need to check Claim \ref{cla: pe} in the proof of Theorem \ref{thm_1}, as that is the only place where SD-Pareto efficiency is used. Recall that by Claim \ref{cla: ir} in the proof of Theorem \ref{thm_1}, for all $j \in \{1,\ldots,k\}$, we have,
    \begin{equation}\label{eq_0}
\rruleof{{j-1}}{\dendow{j}}{P'_N}+\rruleof{j}{\dendow{j}}{P'_N} = 1.
\end{equation}

Now we use ex post Pareto efficiency to show that Claim \ref{cla: pe} holds, that is, for all agents in $\mathcal{C}$, $\varphi(P'_N)$ is the TTC assignment, i.e., for all $j \in \{1,\ldots,k\}$, $\rruleof{{j-1}}{\dendow{j}}{P'_N}= 1$. In view of (\ref{eq_0}), we must consider deterministic assignments that assign the endowments of all agents in $\mathcal{C}$ to the agents in $\mathcal{C}$ only. Moreover, since all agents in $\mathcal{C}$ have different top-ranked alternatives, and these alternatives are the endowments of agents in $\mathcal{C}$, the only Pareto efficient deterministic assignment is when agents in $\mathcal{C}$ are assigned their top-ranked alternatives. By ex post Pareto efficiency, the probabilistic outcome is a convex combination of deterministic Pareto efficient outcomes, thus we must have $\rruleof{{j-1}}{\dendow{j}}{P'_N}= 1$ for all $j \in \{1,\ldots,k\}$, completing the proof of Claim \ref{cla: pe} in the proof of Theorem \ref{thm_1}. This completes the proof of Theorem \ref{thm_3}.
\end{proof}

\section{Proof of Theorem \ref{thm_4}}
\begin{proof}
    Again, we proceed in the lines of the proof of Theorem \ref{thm_1}. As SD-individual rationality is equivalent to ex post individual rationality (Lemma \ref{lem_1}), we only need to check Claim \ref{cl_2} in the proof of Theorem \ref{thm_2}, as that is the only place where SD-pair efficiency is used. We restate Claim \ref{cl_2} here.

    \begin{claim*}
        $\varphi_{kx_1}(\bar{P}_L,P_{-L})=1$.
    \end{claim*}
    \textbf{Proof of the claim} Assume for contradiction that this does not hold, and therefore, by Claim \ref{cl_1}-(i), there exists $i\in \{l,\ldots,k-1\}$ such that $\varphi_{ix_1}(\bar{P}_L,P_{-L})>0$. This means that if we write $\varphi(\bar{P}_L,P_{-L})$ as a convex combination of deterministic assignments, then there must be a deterministic assignment, say $F$, such that $F_{ix_1}=1$. Further, as $\varphi(\bar{P}_L,P_{-L})$ satisfies the condition in Claim \ref{cl_1}-(ii), so does $F$, implying $F_{i+1x_{i+1}}=1$. We show a contradiction by arguing that $F$ is not pair efficient at $(\bar{P}_L,P_{-L})$. To see this, note that in $F$, agent $i$ is assigned $x_1$ and agent $i+1$ is assigned $x_{i+1}$. However, if they swap their assignments, they both will be better off as agent $i$ has $x_{i+1}$ as her top-ranked object and agent $i+1$ prefers $x_1$ over $x_{i+1}$. This completes the proof of the claim, and hence, completes the proof of Theorem \ref{thm_4}.
\end{proof}
    
\end{document}